\documentclass{article}
\title{Corpus-compressed Streaming and the Spotify Problem}
\author{Aubrey Alston (ada2145@columbia.edu)}

\usepackage[left=2.5cm,top=3cm,right=2.5cm,bottom=3cm,bindingoffset=0.5cm]{geometry}

\makeatletter
\renewcommand\paragraph{\@startsection{paragraph}{4}{\z@}%
            {-2.5ex\@plus -1ex \@minus -.25ex}%
            {1.25ex \@plus .25ex}%
            {\normalfont\normalsize\bfseries}}
\makeatother
\usepackage{fancyhdr}
\date{}

\usepackage{amsmath}
\usepackage{amssymb}
\usepackage{MnSymbol}
\usepackage{graphicx}
\usepackage{float}
\usepackage{amsthm}
\usepackage{algorithm}
\usepackage{algpseudocode}
\usepackage{algorithmicx}
\usepackage{program}
\usepackage{caption}

\theoremstyle{definition}
\newtheorem{definition}{Definition}[section]
\newtheorem{theorem}{Theorem}[section]
\newtheorem{lemma}{Lemma}[section]

\begin{document}

\maketitle

\section{Overview}

In this work, we describe a problem which we refer to as the \textbf{Spotify problem} and explore a 
potential solution in the form of what we call \textbf{corpus-compressed streaming schemes}.

Spotify is a digital music streaming service that gives users access to music on demand.
One of its most prominent features is the `playlist' feature: Spotify (and other Spotify users)
maintain and update lists of songs that other users may listen to at any time.  The process of listening 
to these playlists is generally seamless on home networks due to the lack of significant bandwidth 
limitations, but this is not necessarily the case for some mobile users whose mobile broadband networks are 
less than reliable.  The Spotify problem simply names the problem of improving the reliability of streaming playback for users who spend time in these constrained networks.

More generally, the Spotify problem applies in any number of practical domains where devices 
may be periodically expected to experience degraded communication or storage capacity.  One 
obvious solution candidate which comes to mind immediately is standard compression.  Though obviously 
applicable, standard compression does not in any way exploit all characteristics 
of the problem; in particular, standard compression is oblivious to the fact that a decoder 
has a period of virtually unrestrained communication.  Towards 
applying compression in a manner which attempts to stretch the benefit of periods of higher communication 
capacity into periods of restricted capacity, we introduce as a solution the idea of a 
\textbf{corpus-compressed streaming scheme}.  
\newline\newline
This report begins with a formal definition of a corpus-compressed streaming scheme.  Following a 
discussion of how such schemes apply to the Spotify problem, we then give a survey of specific
corpus-compressed scheming schemes guided by an exploration of different measures of description 
complexity within the Chomsky hierarchy of languages. \footnote{This work is the project report produced for the Advanced Algorithms course 
at Columbia University.}
\newline\newline

\tableofcontents

\section{Definition of a Corpus-compressed Streaming Scheme}

We define a \textbf{corpus-compressed streaming scheme} in a setting consisting of two parties, an 
encoder \textit{A} and a decoder \textit{B}.  \textit{A} holds a finite set of 
$n$ distinct binary strings $C=\{c_1,...,c_n\}$, called a \textit{corpus set}, 
in which \textit{B} is interested.  Following a setup phase, \textit{A} has need to convey some 
stream of strings within this corpus set $c_{i_1},c_{i_2},...$ to \textit{B} with 
the following constraints:

\begin{itemize}
\item{During the setup phase, \textit{A} and \textit{B} may communicate freely.}
\item{During the stream, communication between \textit{A} and \textit{B} is costly.}
\item{During the entire exchange, memory is costly for \textit{B}, meaning that 
\textit{B} desires not to store the entire corpus set.}
\item{There is no characterization of the stream beyond the fact that all strings come from 
the corpus set.}
\end{itemize}

\noindent We  give a general definition of a corpus-compressed streaming scheme as well as a 
formal parameterization of the function of such a scheme:  

\theoremstyle{definition}
\begin{definition}{\textit{corpus-compressed streaming scheme}}
Let $C$ be a corpus set containing $n$ distinct strings.  A \textit{corpus-compressed streaming scheme}
(CCSS), is a triplet of algorithms $(Construct(C), Encode(D,x), Decode(D,\hat{x}))$ respectively 
defined as follows:
\begin{enumerate}
\item{$Construct(C)$ takes as input a corpus set and returns a \textbf{schematic object} 
$D$ along with (potentially empty) auxiliary output $A$.  A schematic object is a data structure used to encode and decode elements of the corpus set.}
\item{$Encode(D,A,x)$ takes as input a valid schematic $D$ and associated auxiliary input $A$ 
and a string $x \in C$ and returns 
either an encoding of $x$ or $\perp$ if $x$ is invalid with respect to $D$.}
\item{$Decode(D,\hat{x})$ takes as input a valid schematic $D$ and an encoding $\hat{x}$ and 
returns some string $x \in C$ (where $D$ is a valid schematic for $C$).}
\end{enumerate}
\end{definition}

\theoremstyle{definition}
\begin{definition}{$\delta$-\textit{minimal} $p(n,z)-\epsilon$-\textit{corpus-compressed streaming scheme}}
Let $C$ be a corpus set containing $n$ strings; let $z = \max_{c_i \in C} \lvert c_i \rvert$ be 
the length of the longest string in $C$.  A $\delta$-\textit{minimal} $p(n,z)-\epsilon$-\textit{corpus-compressed streaming scheme}, shorthand $(\delta, p(n,z), \epsilon)$-CCSS, is a corpus-compressed streaming 
scheme $(Construct, Encode, Decode)$ with the following properties:

\begin{enumerate}
\item{\textit{Compression: } For $D = Construct(C)$,
\[\forall c_i \in C, \lvert Encode(D, c_i) \rvert \leq p(n,z)\]

\noindent As defined, the parameter $p(n,z)$ gives the maximum length of the \textbf{streaming code} 
for any individual $c_i \in C$ as a function of the size of the corpus and the maximum element size.}
\item{\textit{Correctness: } For $D = Construct(C)$,

\[ \forall c_i \in C, \frac{ham(Decode(\mathcal{D},Encode(\mathcal{D}, A, c_i)), c_i)}{\lvert c_i \rvert} \leq \epsilon \]

\noindent and where $ham(\cdot, \cdot)$ denotes Hamming distance.  As defined, the parameter $\epsilon$ 
parameterizes the maximum reconstruction error of the scheme.}
\item{\textit{Minimality: } For $D = Construct(C)$, 

\[ \frac{\lvert \mathcal{D} \rvert}{\lvert \mathcal{D}* \rvert} \leq 1 + \delta \] 

\noindent where $\mathcal{D}*$ is the minimum-length satisfactory schematic among some restriction of possible objects $\mathcal{D}$.  As defined, $\delta$ defines the factor by which the output of $Construct(D)$ is 
off from some definition of minimal.  We will see in later sections that this notion of minimality has 
some interesting connections to concepts in algorithmic information theory.}
\end{enumerate}

\end{definition}

\noindent Because this is a first presentation of corpus-compressed streaming, this work 
will explore only schemes guaranteeing exact reconstruction (in other words, we fix $\epsilon$ to 
be 0 in all explorations).  We additionally provide the following trivial lower bound on the 
minimum achievable streaming code length of any CCSS which we will use in our exploration:

\begin{theorem}
\label{CodeLenLB}
Let $S$ be a valid $(\cdot, p(n,z), 0)$ corpus-compressed scheming scheme.  The 
maximum streaming code length for any corpus set must obey the inequality $p(n,z) \geq \log (n+1) - 1$.
\end{theorem}

\begin{proof}
Assume not.  Assume $p(n,z) < \log (n+1) - 1$.  Even if the scheme makes use of variable length 
streaming codes, the scheme may only encode

\begin{align*}
& < \sum_{i=1}^{\log (n+1)-1} 2^i \\
& < 2^{\log (n+1) - 1 + 1} - 1 \\
& < n + 1 - 1\\
& < n
\end{align*}

\noindent distinct strings.  Since $C$ includes $n$ distinct strings but can encode 
only less than that number, we conclude that $S$ cannot possibly be $\epsilon=0$-correct.$\square$.
\end{proof}

\subsection{Application to the Spotify Problem}

We see immediately that the existence of a $(\delta, p(n,z), \epsilon)$-CCSS with reasonably 
sized schematics yields an effective solution to the Spotify problem.  For the sake of illustration,
say that we have a $(\delta, p(n,z), \epsilon)$-CCSS, $(Construct, Encode, Decode)$.  
In the setting of the Spotify problem, we may apply the scheme in the following straight-forward manner:

\begin{enumerate}
\item{Before leaving his or her home network, the user indicates to Spotify that he or she wishes to 
listen to playlist $P=\{s_1,...,s_n\}$.}
\item{Spotify then treats the playlist of songs as a corpus set and sends $D = Construct(P)$ to the user's device.}
\item{The user leaves his or her home network, entering a bandwidth-constrained mobile network.}
\item{In perpetuity, then, when the user requests song $i$, Spotify now sends $Encode(D,s_i)$ (of length 
less than or equal to $p(n,z)$) rather than the entirety of $s_i$.}
\end{enumerate}

Consider the effect such an application would have on the music streaming process.  In essence, the 
user ultimately utilizes his or her temporarily unconstrained bandwidth in the second step to 
receive a reasonably sized data structure that will allow him to stream music at a reduced 
cost in perpetuity; note that by choosing $p(n,z)$ as conservatively as $\frac{z}{r}$, where mobile 
bandwidth is $\frac{1}{r}$-fraction of home bandwidth (often, $r < 2$), 
we would negate any difference in playback he or she might otherwise observe, thus solving the problem.   

\section{Corpus-compressed Streaming Schemes: Concrete Constructions using 
Regular Output Automata}

In this section, we attempt to design corpus-compressed streaming schemes by taking 
inspiration from elements of algorithmic information theory.  Among these elements is the notion of the Kolmogorov complexity of an object, 
one of the most prevalent ideas in algorithmic information theory \cite{CondKolmog}, defined as the length of the shortest Turing machine description which produces said object.

For our purposes, we both project this idea onto the notion of a corpus-compressed streaming scheme 
and direct our analysis to consider another segment of the 
Chomsky hierarchy of languages.  In particular, where Kolmogorov complexity is interested 
in the length of the shortest Turing machine which outputs a string $x$, we would be interested 
in the shortest general output Turing machine $M$ which has the following property for a corpus 
set $C$:

\[ \forall c_i \in C, \exists e_i, \lvert e_i \rvert \leq p(n,z), M(e_i) = c_i \]

We make two observations and derive one question which guides our exploration in this section.
We first note that the decision variant of determining Kolmogorov complexity 
is undecidable, and so there is no algorithmic solution (in the way of $Construct$) capable 
of constructing $M$ given $C$.  We also note that output Turing machines correspond to the 
most encompassing point in the Chomsky hierarchy of languages.  These observations lead 
us to ask the following question: what if instead we restrict schematics to output automata 
corresponding to less encompassing points in the hierarchy?

This question guides the CCSS constructions we derive.  This work 
begins to answer this question by considering corpus-compressed streaming schemes with schematics restricted 
to the least encompassing point in the Chomsky hierarchy:  that of regular languages.
A regular language may be defined as a language which may be recognized by a finite state 
machine; as we are interested in machines with output, in this work we will consider the schemes
we may derive when we restrict schematics to the set of finite state machines with per-state output,
known as Moore machines.  More specifically, we first consider schemes with schematics restricted to 
Moore machines whose underlying acyclic graphs are acyclic and then use these results to make connections 
to the general (cyclic) case.

\subsection{Formalizing our Restriction}

In this section, we are interested in schemes $(Construct, Encode, Decode)$ where 
the output of $Construct(C)$ is a Moore machine defined by the tuple 
$(S,S_0,\Sigma,\Gamma,T,G)$ with the following properties:

\begin{enumerate}
\item{$S_0$ is the unique start state of the machine.}
\item{$\Sigma$, the input alphabet, is the set $\{0,1,\perp\}$.  
$\perp$ is a special end-of-input symbol that is read only at the end of every string.}
\item{$\Gamma = \Sigma$ is the output alphabet.  
In the case of the output alphabet, $\perp$ is a special 
blank symbol which only ever occurs at a starting or final state.}
\item{$T:S\times \Sigma \Rightarrow S$ is a deterministic transition function mapping 
states to successor states given an input symbol.  With respect to 
the transition function of a Moore machine, we use the convention of a 
\textbf{unconditional transition}, which is a transition which is taken 
regardless of whether or not the next input symbol is 0 or 1.  Absent the 
presence of an explicit $\perp$-transition, this transition is taken 
even if there is no next input symbol.}
\item{$G$ is an output function mapping states to their outputs.}
\end{enumerate}

The output of $Construct$ must satisfy the further stipulation that a 
$(\delta, p(n,z), \epsilon)$-CCSS maintain $p(n,z)$-compression 
and $\epsilon$-correctness for the following fixed, universal $Decode$ procedure 
utilized by schemes under this restriction: 

\begin{algorithm}[H]
\caption{Fixed Decoding Procedure}\label{AMMDecode}
\begin{algorithmic}[1]
\Procedure{$Decode$}{$D=(S,S_0,\Sigma,\Gamma,T,G), x$}
\State Beginning at $S_0$, run $D$ on input $x$.  If ever there is a next input symbol but no successor state, return $\perp$.
\State \textbf{return} the sequence of $0$ and $1$ outputs of $D$.
\EndProcedure
\end{algorithmic}
\end{algorithm}

\noindent (Note: while this restriction does not necessarily stipulate a 
universal $Encode$ procedure, the requirement remains that it must exist 
and be efficiently computable.)

Under this restriction of schematics, we define $\delta$-minimality 
with respect to the number of states in a machine.  For a given corpus set $C$,
the minimal schematic is the Moore machine with the smallest number of states of any Moore machine 
satisfying the stated requirements.

In the remainder of this section, we refer to a CCSS under this restriction as a 
\textbf{MM (Moore machine)-restricted CCSS}.  

\subsection{Restricting Schematics to Acyclic Moore Machines}

In \cite{BDDs}, Bryant presented the binary decision diagram (BDD) data structure as a means of
representing and manipulating Boolean functions.  The core mechanism underlying applications 
of BDDs is their function as read-once branching programs: functions are 
represented as rooted, directed acyclic graphs consisting of decision junctions and terminal 
nodes.  Each transition from a decision junction corresponds to a final assignment to exactly 
one variable, and terminal nodes correspond to function evaluations given assignments so far.  
The end result is a graph in which every distinct path corresponds to a distinct variable 
assignment $\vec{x}$ ending with a terminal node having a label corresponding to whether it 
satisfies the formula.

Perhaps more importantly, BDDs are especially useful with respect to Boolean function representation 
because they are amenable to compression through the use of simple reduction rules.  
Given a formula $\Phi$ in $n$ variables and an ordering of those variables, 
there exists the notion of a reduced-ordered BDD 
(ROBDD) which is able to represent every assignment (and its image in $\Phi$) in 
a diagram having often far fewer than $2^n$ nodes for many practical instances.
\newline\newline
Considering once more our interests in this work, there is at least one significant direct parallel between BDDs and the goals of
corpus-compressed streaming schemes.  In particular, consider a Moore machine schematic in which 
there are no cycles of states.  
If we view the graph created by the set of states $S$ and the transition function $T$, 
we see that we also have a rooted, directed acyclic graph with decision junctions at 
states having transition options for both input symbols $0$ and $1$ in which distinct 
paths lead to distinct outputs.

While BDDs and Moore machine schematics are obviously not exactly analogous with respect to 
goals and structure, these observable similarities naturally lead us to question whether 
similar reduction methods in corpus-compressed streaming might yield practical schemes.  
Inspired by this prospect, in this section, we study corpus-compressed streaming in 
our chosen restriction with the additional requirement that the graphs underlying schematics 
be acyclic.  We denote a valid CCSS under this restriction by the term `\textbf{AMM (Acyclic Moore machine)
-restricted CCSS}'.

\subsubsection{An Exact AMM-restricted $(0,\lceil \log n \rceil ,0)$-Corpus-compressed Streaming Scheme}

One of the most powerful reduction rules discussed in \cite{BDDs} is the merging of 
isomorphic subgraphs within non-reduced binary decision diagrams.
We extend and apply this idea on order to derive an exact AMM-restricted $(0, \log n, 0)$-CCSS.
This section proceeds as follows: we (1) present our scheme, (2) provide a proof of this scheme's 
validity, (3) note a fact about the streaming code length achieved, and
(4) provide a worked example.
\newline\newline
\noindent \textit{Scheme 1 } We provide pseudocode for the $Construct$ and $Encode$ procedures of our AMM-restricted CCSS. $Decode$ is given by Algorithm \ref{AMMDecode}.

\begin{algorithm}[H]
\caption{Construct Routine, Scheme 1}\label{AMMConstruct}
\begin{algorithmic}[1]
\Procedure{$Construct$}{$C=\{c_1,...,c_n\}$}
\State STAGE 1: \Comment Construct a MM with a directed tree topology
\State Initialize $D$ as a Moore machine with a single state $S_0$.
\For{$i=1...n$}
\State Set $s=S_0$.
\For{$j=1...\lvert c_i\rvert$}
\State Set $b$ to be the $j$th bit of $c_i$.
\If{$s$ has no $0$-, $1$-, or unconditional transition edges}
\State Add a new state $s'$ to $D$ which outputs $b$.
\State Add an unconditional transition edge from $s$ to $s'$.
\EndIf
\If{$s$ has an unconditional edge to a state $s''$ with output $1-b$}
\State Modify the transition from $s$ to $s''$ to be a transition on 
input symbol $1-b$.
\EndIf
\If{$s$ does not have a transition edge on input symbol $b$}
\State Add a new state $s'$ to $D$ which outputs $b$.
\State Add a transition on input symbol $b$ from $s$ to $s'$.
\EndIf
\State Set $s=s(b)$, where $s(b)$ denotes the state reached following 
a transition on symbol $b$.
\EndFor
\State Add a new node $s'$ with the special $\perp$ output symbol.
\State Add a transition from $s$ to $s'$ on end-of-input symbol $\perp$.
\State Mark $s'$ as a final state.
\EndFor
\State STAGE 2: \Comment Enumerate states in depth order
\State Initialize $Q$ as an empty queue.
\State Initialize $L$ as an empty list of length $\lvert S \rvert$.
\State $Q.enqueue(S_0)$
\State Set $i=0$
\While{$\lvert Q \rvert > 0$}
\State Set $s = Q.dequeue()$
\State Set $L[\lvert S \rvert - i] = s$
\State Set $i = i + 1$.
\State For all states $s'$ such that there is a transition from $s$ to $s'$, 
$Q.enqueue(s')$.
\EndWhile
\algstore{AMMConstruct}
\end{algorithmic}
\end{algorithm}

\begin{algorithm}[H]
\ContinuedFloat
\caption{Construct Routine (Continued)}\label{AMMConstruct}
\begin{algorithmic}
\algrestore{AMMConstruct}
\State STAGE 3: \Comment Apply reduction: merge isomorphic subgraphs
\State Initialize $T$ as an empty associative array (dictionary).
\For{$i=1...\lvert S \rvert$}
\State Set $s=L[i]$.  Let $s.out$ denote the output symbol of state $s$.
\State Set $id=s.out\mid\mid s(0) \mid\mid s(1) \mid \mid s(\perp)$.
\If{$T[id]=null$}
\State Set $T[id]=s$.
\Else
\State Let $r$ be the single node having a transition to $s$.  (Unless $s$ 
is $S_0$, in which case \textbf{break}).
\State Replace the transition from $r$ to $s$ with a transition on the 
same symbol to $T[id]$.
\State Delete $s$ from $D$.
\EndIf
\EndFor
\State If $S_0$ has only an unconditional transition to some state $s'$, 
delete $S_0$ and make $s'$ the start state.
\State \textbf{return} $D$,$A=\emptyset$.
\EndProcedure
\end{algorithmic}
\end{algorithm}

\noindent The $Construct$ routine of scheme 1 begins in stage 1 by constructing a 
Moore machine whose states and transitions take the form of a directed 
binary tree having $n$ paths such that following each path 
produces a unique string in the corpus set.  Starting from a non-output initial state, we iterate through all strings in the corpus set, bit by bit, branching where 
strings following the same path diverge.  In stage 2, the procedure 
prepares for iteration through the states of this Moore machine in order 
of decreasing depth from the start state.  In stage 3, the procedure 
merges isomorphic components in the machine using a light dynamic programming 
approach.  We prove the correctness and properties of this procedure later 
in this section.

The encoding procedure of scheme 1 is given as   

\begin{algorithm}[H]
\caption{Encode Procedure, Scheme 1}\label{AMMDecode}
\begin{algorithmic}[1]
\Procedure{$Encode$}{$D=Construct(C), A, c_i$}
\State If $S_0.out == 1 - c_i[1]$, \textbf{return} error.
\State Set $x$ to be the empty string.
\State Set $s=S_0$.
\For{$j = 1 ... \lvert c_i \rvert$}
\State If $s$ has only an unconditional edge to a state which outputs 
$1-c_i[j]$, \textbf{return error}.
\State If $s$ has no transition to a state which outputs $c_i[j]$, 
\textbf{return error}.
\State If $s$ has at least one non-unconditional transition, 
choose a symbol $b$ which leads to a state which outputs $c_i[j]$ 
and set $x = x \mid \mid b$.
\State Set $s$ to be the successor state of $s$ which outputs $c_i[j]$.
\EndFor
\State If $s$ has no $\perp$-transition, \textbf{return} error.
\State \textbf{return} x
\EndProcedure
\end{algorithmic}
\end{algorithm}

\noindent $Encode$ simply begins at the start state of the machine and
records the transitions along a path in $D$ which outputs the input 
string $c_i$.  We now move to prove that this scheme is indeed a 
$(0, \log{n}, 0)-CCSS$ as well as prove its runtime properties.

\paragraph{Proving $(0,\log{n},0)$-CCSS Validity}

We prove in this section that Scheme 1 is indeed a valid 
$(0, \log{n}, 0)$-CCSS.

\begin{lemma}
\label{ConstructL1}
At the end of stage 1 of $Construct$, $D$ has a tree topology in which the output 
along every path from the root ($S_0$) is unique.
\end{lemma}

\begin{proof}
That the graph underlying $D$ is a directed tree 
is immediate: when a new state is added, it is given a unique parent; likewise, when a 
new state is added, it is never given a transition to an existing state.  Because $D$ is a 
directed tree, there is a unique path from the root to every internal state.  Assume that there exist 
two of these unique paths $P_1$ and $P_2$ each starting with $S_0$ such that the output 
along these paths is equal.  
Because these paths cannot be the same, there must exist a first point of divergence along them. 
But at this point of divergence, there must be two transitions to two distinct states 
having the same output, which is impossible by lines 6-19.$\square$
\end{proof}

\begin{lemma}
\label{ConstructL2}
At the end of stage 1 of $Construct$, there is a path from root to leaf in $D$ 
for each $c_i \in C$ along which the machine will output $c_i$.
\end{lemma}

\begin{proof}
Assume not; assume that there exists
a $c_i$ such that there does not exist a path from $S_0$ to a leaf along 
which 0/1 outputs correspond to $c_i$ at tne end of stage 1.  There must exist a least index $k$ 
from $1$ to $\lvert c_i \rvert$ such that there is a unique (by Lemma \ref{ConstructL1}) 
path from $S_0$ which outputs 
$c_i[1],...,c[k-1]$ (or the empty string if $k=1$) but not $c_i[1],...,c[k]$.  
During iteration $i$ of line $4$, we see 
that the sequence of the first $k-1$ values of $s$ is precisely this path; since transitions 
are never deleted, this path must exist at the end of stage 1.  But we see by lines 6-19 
that a transition to a state with output $c[k]$, meaning that there does exist a path with 
output $c[1],...,c[k]$, a contradiction.$\square$
\end{proof}

\begin{lemma}
\label{ConstructL3}
At the end of stage 1, every path from root to leaf in $D$ outputs a string in 
$C$.
\end{lemma}

\begin{proof}
By construction, every leaf in $D$ has output symbol $\perp$.  There are exactly 
$n$ such leaves added to $D$.  Since paths in $D$ are unique ($D$ forms a tree), there are 
exactly $n$ paths from root to leaf in $D$.  By Lemma \ref{ConstructL2}, there must be a 
path from root to leaf in $D$ which outputs each $c_i \in C$.  Since there are $n$ paths 
among $n$ strings, we conclude that there is a one-to-one correspondence between paths in 
$D$ and strings in $C$.  Every path must therefore output a string in $C$.$\square$
\end{proof}

\begin{lemma}
\label{ConstructL4}
Let $D$ be a Moore machine whose underlying graph is acyclic, rooted at $S_0$, in 
which $\perp$-output states are either $S_0$ or nodes without any outbound transitions (leaves).
For any state $a$ in $D$, let $gen(a)$ be the set of strings which are generated 
by $D_a$, the machine whose underlying graph is the directed acyclic subgraph rooted at 
$a$  (i.e., the set of strings output by $D$ when starting 
at $a$ and following any path to a leaf).
\newline\newline
\noindent For any two distinct states $a$ and $b$, $gen(a) = gen(b)$ if and only if 
there exist states $x \in D_a$, $y \in D_b$ such that $D_x$ is isomorphic to $D_y$ but $x \neq y$.
\end{lemma}

\begin{proof}
$(\Rightarrow)$.  We show that $D_a \cong D_b$ implies $gen(a) = gen(b)$ by an inductive 
argument.  Consider first the case where $a$ and $b$ are leaves.  By our choice of $D$,
$a$ and $b$ must be $\perp$-output leaves.  Therefore $gen(a) = gen(b) = \{\perp\}$, the 
set containing the empty string. Assume now that $D_a \cong D_b \Rightarrow gen(a) = gen(b)$ holds
for the children of some pair of states $x$ and $y$ in $D$ having $D_x \cong D_y$.  Because 
$D_x \cong D_y$, the output bits of $x$ and $y$, $d_x$ and $d_y$, must be equal.  Similarly,
the subgraphs of the $0$-,$1$-,and $\perp$-successors of $x$ must be respectively isomorphic to those of $y$, therefore having the same generated set of strings by our inductive assumption; 
call these $s_0,s_1,s_\perp$ respectively, and let $gen(s_\cdot) = \emptyset$ if the given successor does not exist.

We may explicitly determine $gen(x)$ as 
$gen(x) = \bigcup_{s_\cdot \in s_0, s_1, s_\perp} \{ d_x \mid \mid z, \forall z \in gen(s_\cdot) \}$.  
Likewise, $gen(y) = \bigcup_{s_\cdot \in s_0, s_1, s_\perp} \{ d_y \mid \mid z, \forall z \in gen(s_\cdot) \}$.
Because $d_x = d_y$, then, $gen(x) = gen(y)$.  Since $x \neq y$, we conclude the proof in this direction.
\newline\newline
$(\Leftarrow)$.  Consider now any distinct states $a$ and $b$ with $gen(a) = gen(b) = g$.
If $g$ contains the one-bit string $d_x$, then both $a$ and 
$b$ have $\perp$-successors $s_\perp^{a}$ and $s_\perp^{b}$, both leaves with the same 
set of generated strings.  If $g$ contains 
strings which are 1 in the second bit, $a$ and $b$ must each have exactly one successor with output bit 1, 
$s_1^{a}$ and $s_1^{b}$, each generating the subset of $g$ which is 1 in the second bit.
Likewise, if $g$ contains strings which are 0 
in the second bit, $a$ and $b$ must each have a successor state with output bit 0,
$s_0^{a}$ and $s_0^{b}$, each generating the subset of $g$ which is 0 in the second bit.  (By the same 
argument, if $g$ does not contain strings which are 0 in the second bit, neither $a$ nor $b$ 
have 0-successors; the same holds for strings which are 1 in the second bit.  Thus 
$s_\cdot^{a} \neq \emptyset \iff s_\cdot^{b} \neq \emptyset$).

Now say that we wish to contradict the statement that there exist 
$x \in D_a$, $y \in D_b$ such that $D_x$ is isomorphic to $D_y$ but $x \neq y$.  Then it must be the 
case that $D_{s_1^{a}} \not\cong D_{s_1^{b}} \lor s_1^{a} = s_1^{b}$, 
$D_{s_0^{a}} \not\cong D_{s_0^{b}} \lor s_0^{a} = s_0^{b}$, and 
$D_{s_\perp^{a}} \not\cong D_{s_\perp^{b}} \lor s_\perp^{a} = s_\perp^{b}$.  Indeed, equality cannot 
hold between all successors of $a$ and $b$, else it would be trivially true that $a$ and $b$ are isomorphic, 
and our original statement holds.  There must then exist a $\Delta \in \{0,1,\perp\}$ such that 
$s_\Delta^{a} \neq s_\Delta^{b}$ (whose corresponding subgraphs re not isomorphic).  But we have established that 
$gen(s_\Delta^{a}) = gen(s_\Delta^{b})$, and so we may repeat this argument starting from these 
two states.  Following this pattern, we may continue until (a) all successors are equal, and we 
obtain an example, or (b) are are considering two distinct leaves $x'$, in $D_a$, and $y'$, in $D_b$ 
[we may not consider a leaf and a non-leaf, as it cannot be the case that a leaf and non-leaf generate the same string set].  
But, by choice of $D$, $x'$ and $y'$ have the same output symbol, and so we simultaneously have
$x' \in D_a$, $y' \in D_b$, $x' \neq y'$, and $D_{x'} \cong D_{y'}$.$\square$.
\end{proof}

\begin{lemma}[0-Minimality Condition]
\label{ConstructL5}
Fix a corpus set $C=\{c_1,...,c_n\}$.  Let $X$ be a Moore machine following the 
description given in Lemma \ref{ConstructL4} such that each has exactly $n$ paths 
beginning with $S_0$ and ending at a leaf and where each such path outputs a unique 
string in $C_n$.  Let $Y$ be another Moore machine meeting the same requirements.
Denote by $\lvert X \rvert$ and $\lvert Y \rvert$ respectively the number of states in $X$  
and $Y$.
\newline\newline
\noindent $\lvert X \rvert < \lvert Y \rvert$ implies that there exist states $u$ and 
$v$ such that $D_u \cong D_v$ or states $w$ and $x$ such that a path from 
$S_0$ to $w$ has the same output as one from $S_0$ to $x$.  
As a result, $Y$ is minimal if there are no two distinct 
states $u$ and $v$ in $Y$ such that $D_u \cong D_v$ and no two distinct states $w$ and $x$ 
such that the paths from $S_0$ to each have the same output.
\end{lemma}

\begin{proof}
We provide a proof by contrapositive.  Assume (1) that for all pairs of distinct states 
$u$ and $v$ in $Y$, $D_u \not\cong D_v$. Then by Lemma \ref{ConstructL4}, there do not exist any distinct states $u$ and $v$ in $Y$ having $gen(u) = gen(v)$.  Further, (2) assume that there do not exist 
any distinct states $w$ and $y$ such that any path from $S_0$ to either has the same output.

For the sake of contradiction, assume that $X$ has $\lvert X \rvert < \lvert Y \rvert$.  By the 
pidgeonhole principle, there necessarily exist strings $c_i$ and $c_j$ and indices $k$ and $l$ 
such that state $s_k$ along the path $P_i$ generating $c_i$ in $X$ is state $s_l$ along the path 
$P_j$ generating $c_j$ in $X$.  Without loss of generality, say that $s_k$ outputs the $k$th bit of 
$c_i$ and that $s_l$ outputs the $l$th bit of $c_j$.

\begin{enumerate}
\item{The prefix $t$ of $c_i$ generated by the sub-path of $P_i$ from $S_0$ to $s_k$ in $X$ is also a prefix of 
$c_j$.  But $u_k$ and $v_k$, respectively the states outputing bit $k$ of $c_i$ and $c_j$ 
in $Y$, must necessarily have a path from the root which outputs this prefix.  
This contradicts our choice of $Y$.}
\item{The substring $t_i$ of $c_i$ generated by the sub-path of $P_i$ from $S_0$ to $s_k$ in $X$ is not a 
prefix of $c_j$.  Let $t_j$ be the substring of $c_j$ generated by the sub-path of $P_j$ 
from $S_0$ to $s_l=s_k$ in $X$.  Since every path in $X$ must correspond to a string in 
$C$, we have that $Q_i = \{ t_i \mid \mid q, \forall q \in gen(s_k)\} \subseteq C$ 
and $Q_j = \{ t_j \mid \mid q, \forall q \in gen(s_k)\} \subseteq C$.  

By condition (2) 
on $Y$, there must exist a unique state $u_i$ in $Y$ which generates the substring $t_i$ 
through which all $Q_i$ are generated, 
and there must exist a unique state $v_j$ in $Y$ which generates the substring $t_j$ through 
which all $Q_j$ are generated.  This implies that $gen(s_k) \subseteq gen(u_i)$ and also 
that $gen(s_k) \subseteq gen(v_j)$.  Without loss of generality, say that condition (2) 
holds also for $X$ (given an $X$ for which this is not true, we may simply construct 
one by merging identical prefix paths).  If there exists a string $r \in gen(v_j)$ 
not in $gen(s_k)$, then the string $t_i\mid\mid r$ is in $C$ but cannot be 
generated by $X$ by condition (2); the same argument holds for strings in $gen(u_i)$.  
Indeed, then, it is the case that $gen(v_j) = gen(u_i)$, therefore that 
there exist $a$ and $b$ such that $D_a \cong D_b$ by Lemma \ref{ConstructL4}, a 
contradiction.}
\end{enumerate}
.$\square$.
\end{proof}
\clearpage
\begin{theorem}[Validity of Scheme 1]
Scheme 1 is a valid $(0, \log{n}, 0)$-CCSS.
\end{theorem}

\begin{proof}
\newline \noindent \textit{0-correctness } We demonstrate that scheme 1 is 0-correct.
By lemmas \ref{ConstructL1}, \ref{ConstructL2}, and \ref{ConstructL3}, at the end of 
stage 1 of $Construct$, $D$ has a directed tree topology and has exactly one path 
from root to leaf which outputs one unique string in $C_i$.  In stage 2 of construct ,
we order this tree in decreasing order of the depth of states in $D$.  For each state 
from root to leaf, we then in stage 3 merge isomorphic components 
(which therefore generate the same 
sets of strings).  Because only states generating the same sets of strings are merged, 
$D$ preserves both the number of paths and the 
set of strings generated from the root (the corpus set).   

So long as the string $c_i$ given to $Encode$ is a string in $C$, there exists a path 
in $D$ returned by $Encode$ which outputs $c_i$.  Further, by the construction of 
$D$, there is exactly one path through $D$ for every prefix of any string, and so the 
procedure given in lines 5-10 of $Encode$ will generate $x$ as the sequence 
of non-unconditional inputs needed to generate $c_i$.  It follows directly, then, 
that running $D$ on $x$ (and reading $\perp$ at the end of $x$) will yield $c_i$ 
without any reconstruction error.  It thus holds that 

\[ \forall c_i \in C, \frac{ham(Decode(\mathcal{D},Encode(\mathcal{D}, A, c_i)), c_i)}{\lvert c_i \rvert} \leq \epsilon \]

\noindent \textit{\(\lceil \log{n} \rceil \)-compression } By lemmas \ref{ConstructL1}, \ref{ConstructL2}, 
and \ref{ConstructL3}, and as we have shown in our demonstration of 0-correctness, the 
schematic returned by $Construct$ has exactly $n$ paths.  There can therefore be at most 
$\lceil \log{n} \rceil$ junctions at which there is more than one transition.  At any junction, one 
of the following must be true:

\begin{enumerate}
\item{There is an unconditional transition and a $\perp$ transition.  Then all that is needed 
is a single unary signal indicating to continue.  (Include an additional 1 in the input to continue; 
else end the input where it is.)}
\item{There is a 0-transition and a 1-transition.  Then all that is needed is a single 
binary signal indicating which transition to take.}
\end{enumerate}

\noindent Since in both cases each junction requires only a one-bit indicator, we conclude that 

\[\forall c_i \in C, \lvert Encode(D, c_i) \rvert \leq \log{n} \]

\noindent \textit{0-minimality } Note in the specification of $Construct$ that, 
for every unique prefix among strings in $C$, $D$ has a unique path from 
$S_0$ to some node $x$ which generates that prefix at the end of stage 1.  Thus, for $D$ 
(at the end of stage 1), there do not exist $w$ and $x$ such that a path from 
$S_0$ to $w$ has the same output as one from $S_0$ to $x$.  Because stage 3 does not 
add new states or cause any state to generate a new set of strings (by the argument 
given in our discussion of 0-correctness), we conclude that the final schematic $D$ 
returned maintains this property.

We also claim that the final schematic $D$ returned does not contain any 
distinct states $a$ and $b$ such that $gen(a) = gen(b)$.  We show this by induction on 
the depth of states.  Base case: consider any two distinct states $a$ and $b$ with depth 
greater than or equal to $m$, the maximum depth of any state in $D$. $a$ and $b$ are necessarily 
leaves, thus also necessarily $\perp$-output states.  They will both thus have 
$id=\perp \mid \mid \emptyset \mid \mid \emptyset \mid \mid \emptyset$ in stage 3, line 39, and would have been merged.
As an inductive hypothesis, assume that there are no distinct pairs of nodes $x$ and 
$y$ at depth greater than or equal to $k$ such that $gen(x) = gen(y)$.  Say that there exists a 
pair $a$ and $b$ at or below depth greater than or equal to $k-1$ such that $gen(a) = gen(b)$.
Then $gen(a(0)) = gen(b(0))$, $gen(a(1)) = gen(b(1))$, and $gen(a(\perp)) = gen(b(\perp))$.
This implies that the successors of $a$ respectively generate the same set of strings as the 
successors of $b$.  But since all of these successors are at depth greater than $\geq k$, 
they cannot be distinct by the inductive hypothesis.  
Therefore, where $o$ is the output symbol of $a$ and $b$, $id= o \mid\mid s(0) \mid\mid s(1) \mid \mid s(\perp)$ will be equal 
for these nodes, and so they would therefore be merged in stage 3 of $Construct$.
We note that the procedure given in stage 3 enforces the invariant of this inductive argument 
by processing states in decreasing depth order.

By Lemma \ref{ConstructL4}, then, the schematic $D$ returned by $Construct$ 
does not contain any pairs of states $a$ and $b$ 
such that $D_a \cong D_b$.  Since we have already shown that 
there do not exist $w$ and $x$ such that a path from $S_0$ to $w$ has the 
same output as one from $S_0$ to $x$ in $D$, we conclude by Lemma \ref{ConstructL5} that 
$D$ is minimal, and so 

\[ \frac{\lvert \mathcal{D} \rvert}{\lvert \mathcal{D}* \rvert} \leq 1  \] 

$\square$.
\end{proof}

\paragraph{Runtime Analysis}

In this section, we show the runtime for all procedures (Encode, Decode, Construct)
comprising scheme 1.

\begin{theorem}[Runtime of Scheme 1 Construct]
The $Construct(C=\{c_1,...,c_n\})$ routine of Scheme 1 completes in time
$O(\sum_{i=1}^{n}\lvert c_i \rvert)$. 
\end{theorem}

\begin{proof}
This follows rather directly.  Stage 1 of $Construct$ iterates through all 
bits of all strings in the corpus set.  For each string, it begins at the 
start state of the current Moore machine, and then for each bit, it (a) 
adds or removes at most a constant number of states, (b) modifies a constant 
number of transitions, and (c) makes a single transition.  Stage 1 therefore 
takes $O(\sum_{i=1}^{n}\lvert c_i \rvert)$ time; likewise, $S$ has at most 
$O(\sum_{i=1}^{n}\lvert c_i \rvert)$ states.  Stage 2 is a simple state BFS 
enumeration of the states of $S$, therefore taking time 
$O(\sum_{i=1}^{n}\lvert c_i \rvert)$.  Stage 3 is a pass through the 
nodes in the order of the enumeration determined in stage 2; in each iteration,
an $O(1)$ operation is performed.  Thus, the entire procedure takes time  
$O(\sum_{i=1}^{n}\lvert c_i \rvert)$.

$\square$.
\end{proof}

\begin{theorem}[Runtime of Scheme 1 Encode]
The $Encode(D,C,A,c_i)$ operation of Scheme 1 completes in time 
$O(\lvert c_i \rvert)$.
\end{theorem}

\begin{proof}
This also follows directly.  Note that $Encode$ simply iterates through 
each bit of $c_i$, during each iteration making a constant number of (constant-time)
chehcks, appending at most 1 bit to a single string, and taking at most a 
single transition in $D$.  The entire procedure thus takes time $O(\lvert c_i \rvert)$.

$\square$.
\end{proof}

\begin{theorem}[Runtime of Fixed Decoding Procedure]
The fixed decoding procedure $Decode(D,x)$ completes in time $O(\lvert c_i^*\rvert)$, 
where $c_i^*$ is the string that $x$ encodes if it encodes a string or $x$ if it 
encodes none.
\end{theorem}

\begin{proof}
When $x$ encodes a string, the $Decode$ procedure is simply an execution of $D$ 
which outputs the string $c_i$ which $x$ encodes.  By the structure of Moore machines, 
this thus takes time $O(\lvert c_i\rvert)$.  When $x$ does not encode a string,
some prefix of bits of $x$ are read before the procedure aborts.

$\square$.
\end{proof}

\paragraph{A Note on Optimality and a Worked Example}

We've seen thus far that Scheme 1 is a $(0,\lceil \log{n} \rceil,0)$-CCSS.  From Theorem 
\ref{CodeLenLB}, we know that the minimum possible code length for any 
corpus set and any CCSS is is bounded from below by $\log (n+1) - 1$.  By giving 
Scheme 1 with $p(n,z)=\lvert \log{n} \rvert$, we have therefore provided a CCSS admitting a streaming 
code length only 
additive factor of $\lceil \log{n} \rceil - (\log{(n+1)}-1) \leq \log{\frac{n}{n+1}} + 2$ from optimal.  Noting 
that this factor tends towards $2$ for large $n$ and is strictly less than $2$ for 
all other $n$, we also see that this scheme is nearly optimal with respect to 
streaming code length in a very strong sense.

The remainder of this section provides and illustrates a worked example of Scheme 1 
in order to motivate questions surrounding how to extend the scheme.  Consider the 
corpus set consisting of the vowels in the English alphabet (excluding y):

\[ \mathcal{C} = \{a,e,i,o,u\} \]

\noindent For the sake of exposition, say that each of these vowels are given 
using a naive alphabetical encoding where any letter is represented by its 
ordinal position in the English alphabet.  We may view our corpus set now as

\begin{align*}
\mathcal{C} = \{&a = 1 = &0b00001 \\
&e = 5 = &0b00101 \\
&i = 9 = &0b01001 \\
&o = 15 = &0b01111 \\
&u = 21 = &0b10101 \}
\end{align*} 

\noindent Let us say now that we run $Construct(C)$ (for the $Construct$ procedure 
of Scheme 1) according to the given pseudo-code.  We depict below the Moore 
machine schematic that would be obtained via this procedure:

\begin{figure}[H]
    \centering
    \includegraphics{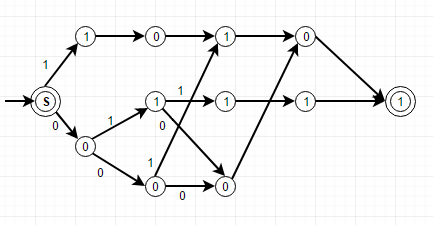}
    \caption{Schematic for C, Output of Scheme 1 Construct}
    \label{fig1}
\end{figure}

As can be seen in Figure \ref{fig1} above, the schematic for $\mathcal{C}$ in Scheme 
1 is an acyclic Moore machine in 11 states.  As promised by the validity of 
scheme 1 as a $(0,\lceil \log{n} \rceil, 0)$-CCSS, we have that we may now convey 
vowels in less than or equal to $\lceil \log{n} \rceil = 3$ bits according to the 
following stream encoding:

\begin{align*}
a &: 000 \\
e &: 001 \\
i &: 010 \\
o &: 011 \\
u &: 1
\end{align*}

\noindent As we see in the case of this example, scheme 1 portrays the properties 
of a CCSS that we desired in order to address the Spotify problem: we see a 40\% 
reduction in the bandwidth required to express vowels under the given naive encoding
without needing to explicitly store the encodings for all vowels.
\newline\newline
The reader may perhaps have noticed at this point that there are yet states in 
\ref{fig1} that may be merged to obtain a smaller schematic whilst maintaining 
the acyclicity property of the restriction.  We illustrate such a merge below:

\begin{figure}[H]
    \centering
    \includegraphics{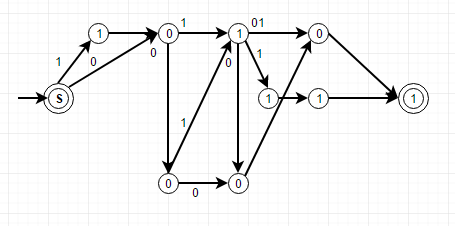}
    \caption{Modified Schematic}
    \label{fig2}
\end{figure}

\noindent Figure \ref{fig2} shows that we may merge two states to obtain 
a smaller schematic (9 states).  Note, however, that this reduction in size 
increases the number of paths through the machine, meaning that stream encodings
must convey more information.  In particular, notice that the above schematic 
contains a path encoding $42=0b101001$, not a vowel (not even a letter 
in our naive encoding); indeed, this schematic 
no longer satisfies $\lceil \log{n} \rceil$-compression, as the longest encoding 
required for any vowel increases from 3 to 5:

\begin{align*}
a &: 000 \\
e &: 00101 \\
i &: 010 \\
o &: 011 \\
u &: 1101
\end{align*}

\noindent While it is certainly true that this schematic no longer satisfies 
the definition of a $(0, \lceil \log{n} \rceil, 0)$-CCSS, this modification shows 
the benefit of increasing $p(n,z)$ for AMM-restricted schemes: if we increase the 
maximum encoding length, we may reduce the size of the schematic.  This phenomenon 
motivates our study of the case where $p(n,z) > \lceil \log{n} \rceil$.

\subsubsection{On the Hardness of Maintaining Minimality for a $p(n,z)>\lceil \log{n} \rceil$ AMM-restricted CCSS}

The previous section portrays very clearly the advantage of increasing maximum 
streaming code length from the near-optimal point of $\lceil \log_{n} \rceil$: 
by doing so, we may reduce the size of the schematic.  Though this is an attractive
prospect, we will show that this is in fact NP-hard to do while maintaining strict
schematic minimality in virtually all interesting cases.  As in the remainder of 
this work, we are considering strictly 0-correct CCSS constructions.

Specifically, we will 
show that (i) it is NP-hard to give an AMM-restricted CCSS maintaining 0-minimality for 
unbounded maximum code length $p(n,z)$ and the stronger result that (ii) 
it is NP-hard to give an AMM-restricted CCSS maintaining 0-minimality for maximum code length
$p(n,z)=\frac{z}{\beta}$ for any fixed $\beta \geq 1$.  We will then extend 
our discussion to consider the maintenance of $\delta$-minimality for general 
$\delta > 0$, formulating a manner in which to relax CCSS constraints 
to allow us to give CCSS constructions which still perform well in practice.

\begin{theorem}[AMM-restriction Hardness for Unbounded Streaming Codes]
\label{UnboundedHardness}
It is NP-hard to give a 0-minimal, 0-correct AMM-restricted CCSS with unbounded 
streaming code length ($p(n,z)=\infty$).
\end{theorem}

\begin{proof}
We show the hardness of this problem via reduction from the well-known NP-hard 
\textit{shortest common supersequence} (SCS) problem for binary alphabets
(shown to be NP-complete in 
\cite{SCSHard}).  The SCS problem is as follows: given a set of $n$ strings composed 
of letters from a fixed (binary, in our case) alphabet, 
determine the shortest possible string $s$ such that each string in the input 
set is a subsequence of $s$.
\newline\newline
The reduction from SCS to AMM-restricted $(0,\infty,0)$-CCSS is direct.  Say 
that we have a $(0,\infty,0)$-CCSS $S=(Construct, Encode, Decode)$.  Consider 
now an instance $I$ of the SCS problem: $I=\{c_1,...,c_n\}$.
As this notation implies, take now $I$ as our corpus set $\mathcal{C}$ and run
$Construct(\mathcal{C}=I)$ to obtain an AMM schematic $D$.  

By the definition 
of an AMM-restricted CCSS, $D$ is a Moore machine whose underlying directed 
graph is acyclic.  Since this graph is acyclic, we may enumerate the states 
of $D$ in topological order $O$.  Take $s$ as the string formed by taking the 
output symbols of each 1- or 0- state in the same order as $O$.  We claim now that 
$s$ is a shortest common supersequence of the original instance $I$, and we show 
this in two parts:

\begin{enumerate}
\item{$s$ is a supersequence of all strings in $I$.  Take any string $c_i \in I$.  
Since $D$ is an AMM schematic 
for $\mathcal{C}=I$, there exists a path of states through which $D$ outputs 
$c_i$.  The $j$th state along this path outputs the $j$th bit of $c_i$, and the 
edges along the path which outputs $c_i$ must obey the topological ordering, meaning 
there is a subsequence of states in the topological ordering $O$ which outputs 
$c_i$.  Taking $s$ as defined, there is therefore a subsequence of bits in $s$ 
which is equal to $c_i$; $s$ is therefore a supersequence of $c_i$. 
}
\item{$s$ is as short as any other supersequence of the strings in $I$.  Assume 
not.  Then there exists a string $s'$ shorter than $s$.  We will use this string 
$s'$ to construct a new Moore machine $D'$: for each bit in $s'$, introduce a new 
state which outputs the same bit.  Next, for each string $c_i \in I$, create a 
path through these states by adding transitions between states $E_j$ and $E_{j+1}$, 
respectively corresponding to the $j$th and the $j+1$th bits of $c_i$, 
for all $j=1,...\lvert c_i\rvert - 1$.  By adding start and end states 
with the special symbol $\perp$ according to the convention we've seen thus far, 
we obtain a valid Moore machine $D'$ which is a valid schematic for $I$.  But $D'$
has as many states as there are bits in $s'$, therefore fewer states than bits 
in $s$ and thus fewer states than $D$, contradicting the 0-minimality of our scheme.
}
\end{enumerate}

$\square$.
\end{proof}

\begin{theorem}[AMM-restriction Hardness for $\frac{z}{\beta}$-bounded Streaming Codes]
\label{BoundedHardness}
It is NP-hard to give a 0-minimal, 0-correct AMM-restricted CCSS with streaming 
code length bounded by $p(n,z)=\frac{z}{\beta}$ for any fixed $\beta \geq 1$.
\end{theorem}

\begin{proof}
Our proof relies on a result by Jiang and Li in \cite{SCSResults} which states 
that it is NP-hard to approximate the SCS problem with a constant approximation 
ratio.  In particular, we show that the existence of a $(0, \frac{z}{\beta}, 0)$ AMM-restricted 
CSS yields a polynomial-time 
$(\beta+1)$-approximation algorithm for the SCS problem for any $\beta \geq 1$; 
the truth of this theorem will then follow from Jiang and Li's result.
\newline\newline
Say that we have a $(0, \frac{z}{\beta}, 0)$ AMM-restricted CSS $S$ 
for some $\beta \geq 1$. Consider an instance $I=\{c_1,...,c_n\}$ of the SCS problem 
(for binary alphabets).  We use $S$ to approximate the shortest common supersequence
of $I$ using the following polynomial-time method:

\begin{enumerate}
\item{Add to $I$ an all-0 string of polynomial length $\beta(z \log_2{\sum_{i=1}^{n}\lvert c_i \rvert} + 1)$, where $z$ is (as before)
the length of the longest string in $I$, obtaining a new set of strings $I'$.}
\item{Run $Construct(I')$ to obtain an AMM schematic $B$.}
\item{Obtain a supersequence $s$ of the strings in $I$ by using the topological 
ordering method from Theorem \ref{UnboundedHardness}, except this time traversing
only the states which are along the output path of some string in the original 
set $I$.  Return $s$.}
\end{enumerate}   

By the proof of Theorem \ref{UnboundedHardness}, we know that $s$ is at least a 
supersequence of all the strings in $I$.  All that remains to show is the optimality
gap between $s$ and the true shortest common supersequence of $I$.  Note that the 
states of the AMM schematic $B$ may be partitioned into a group of states which 
(a) are along the output path for some $c_i \in I$ but not the path for the string 
added to obtain $I'$ or (b) are along the output path for the long string added to 
$I$.  Let the former be $\bar{D}$, and let the latter be $Z$.  Trivially, 
we have $\lvert B \rvert = \lvert \bar{D} \rvert + \lvert Z \rvert$.  

Note now that the length of $s$ corresponds to the number of states 
enumerated in the topological ordering step: this is equal to all states 
along an output path for some $c_i \in I$, including those also along the output path 
for the added string, $\hat{D}$: $\lvert \hat{D} \rvert = \lvert \bar{D} \rvert + \lvert \hat{D} \cap Z \rvert$.  
The optimal SCS $s^*$ corresponds (by our proof of Theorem \ref{UnboundedHardness}) 
to the number of states in the would-be schematic of a AMM-restricted 
$(0, \infty, 0)$-CCSS, $D^*$ given by $Construct(I)$.

Because our CCSS is 0-minimal, we know that $\lvert B \rvert$ is minimized; 
additionally, since $\lvert Z \rvert$ is fixed, we know that $\lvert \bar{D} \rvert$
is minimized.  We note now that $D^*$ may be converted to an AMM schematic for $B'$ by 
simply adding one state per bit of the added string.  This is true because each state 
in $D^*$ has at most $\lvert D^* \rvert \leq \sum_{i=1}^{n}\lvert c_i \rvert$ transitions, meaning that the streaming code gives at most $\log_2{\sum_{i=1}^{n}\lvert c_i \rvert}$ bits of information per transition, therefore that, for the minimum $(0, \infty, 0)$ schematic, at most $z\log_2{\sum_{i=1}^{n}\lvert c_i \rvert}$ bits are needed
to represent any string: since our effective $z$ in $I'$ is $\beta(z\log_2{\sum_{i=1}^{n}\lvert c_i \rvert}+1)$, $z/\beta$ (our maximum encoding length) is strictly larger than this quantity.
Moreover, adding 0-states as mentioned creates exactly one additional choice 
at the start node (hence the $+1$ term), and the resulting schematic $B'$ for $I'$
intersects with precisely no states along the output path for strings in $I$ in 
$D^*$.  

To continue, given that $D^*$ may be used to obtain an AMM schematic for $I'$ in 
which the states along the output paths of strings in $I$  
intersect with no states along the output path corresponding to the added string, 
we know that 
$\lvert \bar{D} \rvert \leq \lvert \lvert D^* \rvert$ by the $0$-minimality 
of our CCSS, thus that 
$\lvert \bar{D} \rvert + \lvert \hat{D} \cap Z \rvert \leq 
\lvert D^* \rvert + \lvert \hat{D} \cap Z \rvert$.  Furthermore, because 
$\lvert Z \rvert = z\beta \leq \beta \lvert D^* \rvert$, we can state 
$\lvert \bar{D} \rvert + \lvert \hat{D} \cap Z \rvert \leq 
\lvert D^* \rvert(\beta + 1)$.  Applying this, we see then that the approximation ratio 
of this algorithm, given by 

\[ \frac{\lvert \bar{D} \rvert + \lvert \hat{D} \cap Z \rvert}{ \lvert D^* \rvert} \]

\noindent is in fact 

\[ \leq \frac{\lvert D^* \rvert(\beta + 1)}{ \lvert D^* \rvert} = \beta + 1 \]

$\square$.
\end{proof}

\paragraph{Considering $\delta > 0$-Minimality}

We saw in Theorems \ref{UnboundedHardness} and \ref{BoundedHardness} that it 
is NP-hard to maintain 0-minimality in an AMM-restricted CCSS when the streaming 
code length is either unbounded or bounded as a constant fraction of the length 
of the longest string in a corpus set.  A natural question to ask now is this:
does the problem remain hard when we relax the $\delta$ parameter?  I.e.,
is it hard to maintain $\delta$-minimality when varying code length bound?

We give a result in this section which indicates that the answer to this question 
is yes.  We specifically rely on yet another result from 
\cite{SCSResults} to show that there exists a natural floor to the maximum value 
of $\delta$ below which unexpected complexity-theoretic results are implied.

\begin{theorem}[Minimality Floor of AMM-restricted Schemes]
\label{MinimalityRoof}
For every $\beta \geq 1$, there exists an $\alpha$ such that the existence of a 
$(\log^\alpha{n}, \frac{z}{\beta},0)$-CCSS implies that 
$NP \subseteq DTIME(2^{polylog(n)})$.
\end{theorem}

\begin{proof}
Our proof relies on the following result of Jiang and Li \cite{SCSResults}:
there exists a constant $\gamma > 0$ such that, if SCS has a polynomial-time 
approximation algorithm with ratio $O(\log^\gamma{n})$, then 
$NP \subseteq DTIME(2^{polylog(n)})$.  We then need only to show that, for 
generic $\beta$, the existence of a $(f(n), \beta, 0)$-CCSS implies the 
existence of an $O(f(n)+\beta)$-ratio SCS approximation algorithm.  The 
theorem then follows directly from the result of Jiang and Li.
\newline\newline
That a $(f(n), \beta, 0)$-CCSS yields an $O(f(n)+\beta)$-ratio SCS approximation 
algorithm follows almost exactly according to the argument given in the proof 
of Theorem \ref{BoundedHardness} with the following exception: where $\bar{D}$ 
is the set of states along the output paths of the strings of $I$ and 
$D^*$ is the set of states in the minimal AMM schematic, we use the 
fact that $\lvert \bar{D} \rvert \leq \lvert D^* \rvert$ to instead obtain 
that $f(n) \lvert \bar{D} \rvert \leq f(n) \lvert D^* \rvert$.  Introducing 
$D'$ as the corresponding set of states for the schematic returned by the 
$f(n)$-minimal scheme, we can (using the same strategy in the proof of 
Theorem \ref{BoundedHardness}) show that 
$\lvert \bar{D}' \rvert + \lvert \hat{D}' \cap Z \rvert 
\leq f(n) \lvert D^* \rvert + \beta \lvert D^* \rvert$.  The approximation ratio 
is immediate.

$\square$.
\end{proof}

\paragraph{Relaxing the Problem: from $f(\cdot)$-Minimality to $\mathbb{E}[f(\cdot)]$-Minimality}

This section overall has shown that analysis of AMM-restricted 
CCSS constructions through a strict interpretation of minimality 
reveals very unyielding hardness results.  Theorems \ref{UnboundedHardness}
and \ref{BoundedHardness} show us that it is hard to give CCSS constructions 
bounding code length as a fraction of maximum string length; indeed, this hardness
will present itself once more if we instead consider code-length bounds 
$p(n,z)$ as a function of $n$ on instances where $n << z$, and so the prospect of
exploring this possibility seems equally grim.  Theorem \ref{MinimalityRef} shows
us that there is a poly-logarithmic minimality floor under which the prospect of 
giving efficient schemes seems unlikely; this is especially daunting when we 
consider that, in practice, we would like to see $\delta$-minimality for small, constant 
values of $\delta$, e.g. 1.25 or 1.5.

More positively, however, the structure of the problem of giving AMM-restricted 
CCSS constructions and how closely related it seems to the SCS problem 
does show some promise for relaxations which still allow resulting 
constructions to be useful in practice.  In particular, \cite{SCSResults} notes 
that simple greedy algorithms for the SCS problem yield solutions which 
are bad in the worst case but are nearly optimal in the expected case.  Toward 
the end of exploiting this similarity, we suggest the exploration of a new class
of AMM-restricted CCSS constructions: constructions in which minimality is viewed 
in the expected case.  Where we previously were married to the idea of giving an 
AMM-restricted $(f(\cdot), \cdot, \cdot)$-CCSS which is strictly $f(\cdot)$-minimal 
in all cases, we suggest considering AMM-restricted $(\mathbb{E}[f(\cdot)], \cdot, \cdot)$ schemes which is $f(\cdot)$-minimal in the \textbf{expected case}.  If the 
structural similarities with SCS hold to this point, it may be possible that 
greedy approaches perform well in the expected case.

Due to the scope of this report and in the interest of being able to touch 
a broader range of issues related to MM-restricted CCSS constructions in 
general, we leave both the formulation of $\mathbb{E}[f(\cdot)]$-minimality 
and the presentation of algorithms fitting its definition open. 


\subsection{Restricting Schematics to General Moore Machines}

The reader may recall from section 3.2.1.3 that we saw a method of reducing 
AMM schematic size by merging states and introducing new paths.  In the rest of this 
report, we do away with the acyclicity requirement.  Now that we have the 
ability to introduce cycles into the graph underlying our schematics, we may 
in fact reduce schematic size even further.  Anecdotally, we illustrate how 
the addition of a single self-loop reduces the size of the schematic depicted 
in Figure 2 from 9 to 7:

\begin{figure}[H]
    \centering
    \includegraphics{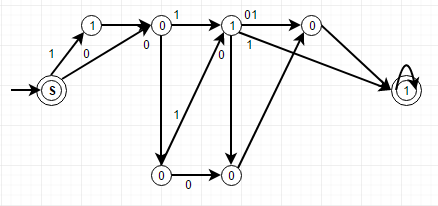}
    \caption{Schematic with Loop}
    \label{fig3}
\end{figure}

In truth, the ability to include cycles in our schematics does more than give 
us the ability to reduce schematic size: not even the same hardness properties no 
longer hold.  For example, Theorem \ref{UnboundedHardness} showed that 
schematic minimization is NP-hard for unbounded code lengths in the AMM-restricted
case; the relaxed MM-restriction contradicts this directly, as invariably 
something similar to one of the following schematics will be optimal:

\begin{figure}[H]
    \centering
    \includegraphics{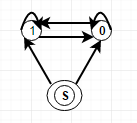}
    \caption{Cyclic Schematic for Degenerate Case, Unbounded Code Length}
    \label{fig4}
\end{figure}

\noindent Above, we see that an NP-hard instance of CCSS construction problem in the 
AMM restriction becomes the degenerate case for the general MM restriction.  
Owing to potentially significant savings in schematic size and this 
significant gap in hardness, we in this section begin to explore the prospect of 
MM-restricted CCSS constructions.  In the interest of adhering to the expected 
scope of this report, this exploration is limited to an informal presentation of 
an equivalent formulation meant to motivate future study of the problem.

\subsection{An Equivalent Formulation of the Construction Problem}

We denote by the \textit{construction problem} the task of designing the 
$Construct$ algorithm of a CCSS.  In this section, we informally present an 
equivalent formulation of the MM-restricted construction problem as a bicriteria 
partitioning problem in directed graphs.  We then use this formulation as a 
basis of discussing directions for continued work.
\newline\newline
As stated, we equivalently pose the MM-restricted construction problem as a 
partitioning problem in directed graphs.  The input for such a problem is a 
directed acyclic graph $F=(V,E)$, an associated labeling function $L(u)$ which 
maps each node in $V$ to a color in $\{0,1\}$, a natural number parameter $k$, 
a natural number parameter $\tau$,  and a set of \textit{trace-paths} 
$P_1,...,P_m$ through $F$.  The output desired by this problem is a disjoint 
monochromatic (as in each partition is monochromatic) 
$k$-partitioning $G_1,...,G_k$ of $F$ which satisfies $\tau$ in the following specific sense:

\begin{itemize}
\item{For each trace path $P_i$, define the \textit{partition-path} of 
$P_i$, $p(P_i)$, to be the sequence of partitions reached along the path $P_i$.}
\item{For each partition $G_i$, define the degree of $G_i$, $deg(G_i)$ to be the 
number of distinct partitions $G_j$ such that $G_j$ is the successor of $G_i$ 
along some partition path.}
\item{We say that a partition $G_1,...,G_k$ \textit{satisfies} $\tau$ if 
and only if 
$\max_{P_i} \sum_{j=1}^{\lvert P_i \rvert} \lceil \log{deg(p(P_i)[j])} \rceil \leq \tau$.}
\end{itemize}  

We omit a formal proof of the correspondence between this formulation and the 
construction problem for the sake of maintaining the scope of this report, 
but we summarize it informally.  In the above formulation, $k$ corresponds to 
the number of states in a schematic, and the parameter $\tau$ corresponds 
to the maximum streaming code length.  The quantity which must satisfy $\tau$ corresponds precisely to the length of the longest streaming code.  
Thus, to give an MM-restricted 
$(0, p(n,z), 0)$-CCSS, we would solve the given problem for the minimum $k$ yielding
a partition satisfying $p(n,z)$.  Again, we note that we leave the formal proof of this correspondence open.
\footnote{Indeed, the problem which corresponds exactly requires a slight modification 
in that the input graph must support having a constant number of unlabeled nodes, 
but we leave this out for ease of exposition.}
\newline\newline
This formulation is useful because it gives us a natural framework in which to 
develop schematic minimization algorithms, but it is also useful for the sake 
of analyzing hardness as we did for the AMM-restricted schemes.  We can 
see this, for example, even in an ability to view hardness of the AMM-restricted case through the 
lens of this formulation: if we impose a partial acyclicity constraint, 
namely dictating that there must not exist edges with start- and end-points in the same partition, 
the above formulation reduces to the problem of determining the chromatic 
number of an arbitrary undirected graph (when $\tau$ is unbounded).

Focusing once more on the MM-restricted case, our hope is that this formulation 
makes it possible to analyze the limits of MM-restricted schemes in a natural way.
We leave the problem of this analysis open, but we note that other similar 
partitioning problems \cite{PartitionHard} are both hard to solve exactly 
and approximately.  We additionally hope that this formulation simplifies the 
task of developing specific MM-restricted CCSS constructions, but we leave also 
this problem open.

\section{Future Work}

We have in this work explored the surface of corpus-compressed streaming 
as a solution to the Spotify problem, substantiating a particular strategy 
which utilizes regular function automata as schematics.  Even with respect to 
this strategy, our concrete results are limited to the development of a single
AMM-restricted $(0, \lceil \log{n} \rceil, 0)$-CCSS and a few hardness results for 
this specific restriction.  While our discussion has touched on extensions of this 
strategy, our work nonetheless leaves open quite a few questions which seem 
worthwhile to explore in the future.  We outline these areas for future work below.

\begin{itemize}
\item{\textit{Lossy Reconstruction } All of our schemes, hardness results, and 
discussion thus far have been restricted to consideration of lossless schemes 
having the $\epsilon$ parameter fixed at 0.  Future work could consider 
developing new lossless CCSS constructions, extending the AMM-restricted 
scheme we've developed to make strategic use of loss, or even seeking to extend or 
refine our hardness results in the lossless case.}
\item{\textit{Expected Minimality } In our presentation of results for AMM-restricted 
schemes, we noted the possible utility of developing a notion of expected-case 
minimality rather than strict minimality of a CCSS.  Future work could look at 
giving a precise formulation of this notion and, moreover, using it to 
construct practically useful schemes with $p(n,z)$ varied away from 
$\lceil \log{n} \rceil$.}
\item{\textit{Taking Other Parameters in Expectation } Related to the previous point, it 
seems natural to also consider the benefit of also taking the maximum streaming code 
length and loss parameters in expectation.}
\item{\textit{MM-restricted Schemes } The last section gave an overview of the 
benefits of studying MM-restricted constructions.  Future work in this area could 
look at formally proving properties of these constructions (potentially using the 
equivalent formulation we have provided) or, perhaps more importantly, 
developing MM-restricted schemes.}
\item{\textit{Schemes from Other Function Automata } This work has focused 
exclusively on constructing schematics from regular function automata.  While 
the hardness results we found at this level of the hierarchy may be a 
deterrent from any attempt at operating at a higher level, studying what happens 
when we do could nonetheless be worthwhile.}
\item{\textit{Other Restriction Constraints } Related to the previous, it would 
of course be worthwhile to consider other means of restricting schematics independent 
of language considerations.  One particular area of interest could be exploring 
restrictions of schematics to auto-encoders with specific properties.}
\item{\textit{Practical Evaluation of Schemes } Outside of the realm of theory, one
major area of future work concerns the implementation and evaluation of CCSS 
constructions on real-world data and in real-world environments.}
\end{itemize}

\section{Conclusion}

In this work, we have defined the \textbf{Spotify problem} and explored a 
potential solution in the form of a new algorithmic goal which 
we refer to as \textbf{corpus-compressed streaming schemes}.  After formally substantiating 
the notion of a corpus-compressed streaming scheme, we explored a specific strategy of 
constructing them based upon a Kolmogorov-like use of regular function automata as an `almost self-extracting'
archive.  Following a presentation of results including a concrete, nearly optimal 
scheme (under a specific restriction) and hardness properties, we further motivate and outline 
opportunities for future work in this area.

\bibliographystyle{acm}
\bibliography{proposal}

\begin{thebibliography}{1}

\bibitem{BDDs}
{\sc Bryant, R.~E.}
\newblock Graph-based algorithms for boolean function manipulation.
\newblock {\em IEEE Trans. Comput. 35}, 8 (Aug. 1986), 677--691.

\bibitem{PartitionHard}
{\sc Feldmann, A.~E.}
\newblock {\em Fast Balanced Partitioning Is Hard Even on Grids and Trees}.
\newblock Springer Berlin Heidelberg, Berlin, Heidelberg, 2012, pp.~372--382.

\bibitem{SCSResults}
{\sc Jiang, T., and Li, M.}
\newblock {\em On the approximation of shortest common supersequences and
  longest common subsequences}.
\newblock Springer Berlin Heidelberg, Berlin, Heidelberg, 1994, pp.~191--202.

\bibitem{CondKolmog}
{\sc Li, M., and Vitnyi, P.~M.}
\newblock {\em An Introduction to Kolmogorov Complexity and Its Applications},
  3~ed.
\newblock Springer Publishing Company, Incorporated, 2008.

\bibitem{SCSHard}
{\sc Räihä, K.-J., and Ukkonen, E.}
\newblock The shortest common supersequence problem over binary alphabet is
  np-complete.
\newblock {\em Theoretical Computer Science 16}, 2 (1981), 187 -- 198.

\end{thebibliography}

\end{document}